\documentclass{amsart}

\pdfoutput=1

\usepackage[utf8x,utf8]{inputenc} 
\usepackage{hyperref}
\usepackage{graphicx}
\usepackage{todonotes}

\newtheorem{Thm}{Theorem}
\newtheorem{Lem}[Thm]{Lemma}

\newcommand{\pN}[1]{\textsc{#1}}
\newcommand{\MaxMotif}{\pN{Maximum Graph Motif}}
\newcommand{\ClosestMotif}{\pN{Closest Graph Motif}}
\newcommand{\SetCover}{\pN{Set Cover}}

\newcommand{\defproblem}[3]{
  \vspace{1mm}
\noindent\fbox{
  \begin{minipage}{0.95\textwidth}
  #1 \\
  {\bf{Input:}} #2  \\
  {\bf{Question:}} #3
  \end{minipage}
  }
  \vspace{1mm}
}

\begin{document}

\title[Constrained Multilinear Detection]{Constrained Multilinear Detection\\and Generalized Graph Motifs*}

\thanks{*A preliminary conference abstract of this work has appeared as A. Bj\"orklund, P. Kaski, and \L. Kowalik, ``Probably optimal graph motifs,'' Proceedings of the 30th International Symposium on Theoretical Aspects of Computer Science (STACS 2013, Kiel, February 27--March 2, 2013), Leibniz International Proceedings in Informatics 20, Schloss Dagstuhl–Leibniz-Zentrum für Informatik, 2013, pp.~20--31.}

\author{Andreas Bj\"orklund}
\address{Andreas Bj\"orklund, Department of Computer Science, Lund University, Sweden}
\email{andreas.bjorklund@yahoo.se}
\author{Petteri Kaski}
\address{Petteri Kaski, Helsinki Institute for Information Technology HIIT \& 
Department of Information and Computer Science, Aalto University, Finland}
\email{petteri.kaski@aalto.fi}
\author{\L{}ukasz Kowalik}
\address{\L{}ukasz Kowalik, Institute of Informatics, University of Warsaw, Poland}
\email{kowalik@mimuw.edu.pl}

\begin{abstract}
We introduce a new algebraic sieving technique to detect constrained multilinear monomials in multivariate polynomial generating functions given by an evaluation oracle. As applications of the technique, we show an $O^*(2^k)$-time polynomial space algorithm for the $k$-sized {\sc Graph Motif} problem.
We also introduce a new optimization variant of the problem, called {\sc Closest Graph Motif} and solve it within the same time bound.
The {\sc Closest Graph Motif} problem encompasses several previously studied optimization variants, like {\sc Maximum Graph Motif}, {\sc Min-Substitute Graph Motif}, and {\sc Min-Add Graph Motif}. Finally, we provide a piece of evidence that our result might be essentially tight: the existence of an $O^*((2-\epsilon)^k)$-time algorithm for the {\sc Graph Motif} problem implies an $O((2-\epsilon')^n)$-time algorithm for {\sc Set Cover}.
\end{abstract}

\maketitle

\section{Introduction}

Many hard combinatorial problems can be reduced to the framework of detecting whether a multivariate polynomial $P(\vec x)=P(x_1,x_2,\ldots,x_n)$ has a monomial with specific properties of interest. In such a setup, $P(\vec x)$ is not available in explicit symbolic form but is implicitly defined by the problem instance at hand, and our access to $P(\vec x)$ is restricted to having an efficient algorithm for computing values of $P(\vec x)$ at points of our choosing. This framework was pioneered by Koutis~\cite{koutis-icalp08}, Williams~\cite{williams-ipl}, and Koutis and Williams~\cite{koutis-williams-icalp09} for use in the domain of parameterized subgraph containment problems, and it currently underlies the fastest known parameterized algorithms for many basic tasks such as path and packing problems~\cite{BHKK-narrow-sieves}. 

The present paper is motivated by recent works of Guillemot and Sikora~\cite{gs10} and Koutis~\cite{koutis-ipl}, who observed that functional motif discovery problems in bioinformatics are also amenable to efficient parameterized solution in the polynomial framework. Following Koutis~\cite{koutis-ipl}, applications in this domain require one to detect monomials in $P(\vec x)$ that are both {\em multilinear} and further {\em constrained} by means of colors assigned to variables $\vec x$, so that the combined degree of variables of each color in the monomial may not exceed a given maximum multiplicity for that color. Our objectives in this paper are to (i) present an improved algebraic technique for constrained multilinear detection, (ii) generalize the technique to allow for approximate matching at cost, and (iii) derive improved algorithms for graph motif problems, together with evidence that our algorithms may be optimal in the exponential part of their running time. We also introduce a new common 
generalization---the {\em closest graph motif problem}---that tracks the weighted edit distance between the target motif and each candidate pattern; this in particular generalizes both the minimum substitution and minimum addition variants of the graph motif problem introduced by Dondi, Fertin, and Vialette~\cite{DFV11-cpm}.

Let us now describe our main results in more detail, starting with algebraic contributions and then proceeding to applications in graph motifs. All the algebraic contributions rely essentially on what can be called the ``substitution-sieving'' method in characteristic 2 \cite{bjorklund-hamilton,BHKK-narrow-sieves}. 

\subsection{Multilinearity}
\label{sect:basic}

To ease the exposition and the subsequent proofs, it will be convenient to start with a known, non-constrained version of the substitution sieve that exposes multilinear monomials.

Let $P(\vec x)=P(x_1,x_2,\ldots,x_n)$ be 
a multivariate polynomial 
over a field of characteristic 2 
such that every monomial $x_1^{d_1}x_2^{d_2}\cdots x_n^{d_n}$ has 
total degree $d_1+d_2+\ldots+d_n=k$. A monomial is {\em multilinear} if 
$d_1,d_2,\ldots,d_n\in\{0,1\}$.

For an integer $n$, let us write $[n]=\{1,2,\ldots,n\}$. 
Let $L$ be a set of $k$ labels. 
For each index $i\in[n]$ and label $j\in L$, 
introduce a new variable $z_{i,j}$. 
Denote by $\vec z$ the vector of all variables $z_{i,j}$.

\begin{Lem}[Non-constrained multilinear detection \cite{bjorklund-hamilton,BHKK-narrow-sieves}]
\label{lem:basic}
The polynomial $P(\vec x)$
has at least one multilinear monomial if and only if the polynomial 
\begin{equation}
\label{eq:unconstrained-sieve}
Q(\vec z)=
\sum_{A\subseteq L}P\bigl(z_1^A,z_2^A,\ldots,z_n^A\bigr)
\end{equation}
is not identically zero, where $z_i^A=\sum_{j\in A} z_{i,j}$
for all $i\in[n]$ and $A\subseteq L$.
\end{Lem}

\noindent
{\em Remark.} We can now observe the basic structure of the sieve \eqref{eq:unconstrained-sieve}: by making $2^k$ substitutions of the new variables $\vec z$ into $P(\vec x)$, we reduce the question of existence of a multilinear monomial in $P(\vec x)$ into the question whether the polynomial $Q(\vec z)$ is not identically zero. The latter can be tested probabilistically by one {\em evaluation} of $Q(\vec z)$ at a random point, which reduces via \eqref{eq:unconstrained-sieve} into {\em evaluations} of $P(\vec x)$ at $2^k$ points. This will be the basic structure in all our subsequent algorithm designs.

\subsection{Constrained multilinearity}
\label{sect:constrained}

We are now ready to state our main algebraic contribution.
Let $C$ be a set of at most $n$ {\em colors} such that each
color $q\in C$ has a {\em maximum multiplicity} $m(q)\in\{0,1,\ldots,n\}$.
Associate with each index $i\in[n]$ a color $c(i)\in C$. 
Let us say that a monomial $x_1^{d_1}x_2^{d_2}\cdots x_n^{d_n}$ 
is {\em properly colored} if the number of occurrences of each color 
is at most its maximum multiplicity, or equivalently, 
for all $q\in C$ it holds that $\sum_{i\in c^{-1}(q)}d_i\leq m(q)$.

Associate with each color $q\in C$ 
a set $S_q$ of $m(q)$ {\em shades} of the color $q$, such that $S_q$ 
and $S_{q'}$ are disjoint whenever $q\neq q'$. Let $S=\cup_{q\in C}S_q$.

For each index $i\in[n]$ and each shade $d\in S_{c(i)}$, introduce 
a new variable $v_{i,d}$. 
For each shade $d\in S$ and each label $j\in L$, introduce 
a new variable $w_{d,j}$.

\begin{Lem}[Constrained multilinear detection]
\label{lem:constrained}
The polynomial $P(\vec x)$
has at least one monomial that is both multilinear and properly colored
if and only if the polynomial 
\begin{equation}
\label{eq:constrained-sieve}
Q(\vec v,\vec w)=
\sum_{A\subseteq L}P\bigl(u_1^A,u_2^A,\ldots,u_n^A\bigr)
\end{equation}
is not identically zero, where 
$u_i^A=\sum_{j\in A} u_{i,j}$
and 
$u_{i,j}=\sum_{d\in S_{c(i)}}v_{i,d}w_{d,j}$
for all $i\in[n]$, $j\in L$, and $A\subseteq L$.
\end{Lem}

\noindent
{\em Remark.} This lemma enables us to (probabilistically) detect a constrained multilinear monomial of degree $k$ using $2^k$ evaluations of $P(\vec x)$, assuming that we are working over a sufficiently large field of characteristic 2. This solves an open problem posed by Koutis at a Dagstuhl seminar in 2010~\cite{koutis-dagstuhl}, and forms the core of our algorithm in Theorem~\ref{thm:max-motif}.

\subsection{Cost-constrained multilinearity}
\label{sect:cost-constrained}

The previous setting admits a generalization where we associate {\em costs}
to decisions to arrive at a proper coloring. 
Accordingly, we assume that no coloring $c:[n]\rightarrow C$ has 
been fixed {\em a priori}, but instead associate with each 
index $i\in [n]$ and each color $q\in C$ a nonnegative 
integer $\kappa_i(q)$, the {\em cost} of assigning the color $q$ to $i$. 

Once a coloring $c:[n]\rightarrow C$ has been assigned, 
the {\em cost} of a monomial $x_1^{d_1}x_2^{d_2}\cdots x_n^{d_n}$
in the assigned coloring is $\sum_{i\in[n]} d_i\kappa_i(c(i))$. 
The objective now becomes to detect a multilinear monomial that has 
the minimum cost under a proper coloring.

For each index $i\in[n]$ and each shade $d\in S$, introduce 
a new variable $v_{i,d}$. 
For each shade $d\in S$ and each label $j\in L$, introduce 
a new variable $w_{d,j}$. Introduce a new variable $\eta$.

\begin{Lem}[Cost-constrained multilinear detection]
\label{lem:cost-constrained}
The polynomial $P(\vec x)$
has at least one monomial that is both multilinear and admits
a proper coloring with cost $\sigma$ 
if and only if the polynomial 
\begin{equation}
\label{eq:cost-constrained-sieve}
Q(\vec v,\vec w,\eta)=
\sum_{A\subseteq L}P\bigl(u_1^A,u_2^A,\ldots,u_n^A\bigr)
\end{equation}
has at least one monomial whose degree in the variable
$\eta$ is $\sigma$, where 
$u_i^A=\sum_{j\in A} u_{i,j}$
and 
\begin{equation}
\label{eq:cost-u}
u_{i,j}=\sum_{q\in C}\eta^{\kappa_i(q)}\sum_{d\in S_q}v_{i,d}w_{d,j}
\end{equation}
for all $i\in[n]$, $j\in L$, and $A\subseteq L$.
\end{Lem}

\noindent
{\em Remark.} The previous lemma may be extended to track multiple
cost parameters $\eta_1,\eta_2,\ldots$ simultaneously. 
In fact, this will be convenient
in our algorithm underlying Theorem~\ref{thm:closest-motif}. We also
observe that in applications one typically works with a (random)
evaluation in the variables $\vec v$ and $\vec w$, but seeks to 
recover an explicit polynomial in $\eta$ as the output of the sieve,
typically by a sequence of evaluations at distinct points,
followed by interpolation to recover the polynomial in $\eta$.

\subsection{Graph motif problems}
The application protagonist for our algebraic tools will be the
following problem and its generalization.

\medskip
\defproblem{%
\MaxMotif{}~\cite{DFV09-cpm}%
}{%
A connected, undirected host graph $H$ with $n$ vertices and $e$ edges, 
a multiset $M$ of colors over a base color set $C$,
a coloring $c:V(H)\rightarrow C$ for the vertices of $H$, 
and a positive integer $k$.
}{%
Is there a subset $K\subseteq V(H)$ of size $k$ such that 
(a) the subgraph induced by $K$ in $H$ 
is connected, and 
(b) the multiset $c(K)$ of colors is a subset of $M$, taking
multiplicities into account?
}
\medskip

\noindent
{\em Background.} Graph motif problems were introduced by Lacroix et al.~\cite{LacroixFS06} and motivated by applications in bioinformatics, specifically in metabolic network analysis. The \MaxMotif{} problem was introduced by Dondi, Fertin, and Vialette~\cite{DFV09-cpm}. It is known to be NP-hard even when the given graph is a tree of maximum degree 3 and each color may occur at most once~\cite{fellows-jcss}. However, in practice the parameter $k$ is expected to be small, what motivates the research on so-called FPT algorithms parameterized by $k$, that is, algorithms with running times bounded from above by a function $f(k)$ times a function polynomial in the input size, which is commonly abbreviated by $O^*(f(k))$. Indeed, Fellows et al.~\cite{Fellows-ICALP} discovered that such an algorithm exists, which was followed by a rapid series of improvements to $f(k)$ \cite{Fellows-ICALP,cpm,gs10}, culminating in the $O^*(2.54^k)$-time algorithm of Koutis~\cite{koutis-ipl} (see Table~\ref{tbl:history}).

\begin{table}[ht]
 \begin{center}
\begin{tabular}{lll}
\hline Paper & Running time & Approach\\\hline
Fellows et al.~\cite{Fellows-ICALP} & $O^*(87^k)$, implicit & Color-coding\\
Betzler et al.~\cite{cpm} & $O^*(4.32^k)$ & Color-coding\\
Guillemot and Sikora~\cite{gs10} & $O^*(4^k)$ & Multilinear detection\\
Koutis~\cite{koutis-ipl} & $O^*(2.54^k)$ & Constrained multilinear detection\\
this work & $O^*(2^k)$ & Constrained multilinear detection\\\hline
\end{tabular}
\end{center}
\caption{Progress on FPT algorithms for the $k$-sized graph motif problem}
\label{tbl:history}
\end{table}

From a high-level prespective the two key ideas underlying our main theorem in this section are (i) an observation of Guillemot and Sikora~\cite{gs10} that {\em branching walks}~\cite{nederlof-steiner} yield an efficient polynomial generating function for connected sets, and (ii) Lemma~\ref{lem:constrained} that builds on work by Koutis~\cite{koutis-ipl}.

\medskip
\noindent
{\em Our results.} 
The coefficient $\mu=O(\log k\log \log k\log \log \log k)$ in the following theorem reflects the time complexity of basic arithmetic 
(addition, multiplication) in a finite field of size $O(k)$
and characteristic 2 \cite{algebraic-complexity-theory}.

\begin{Thm}
\label{thm:max-motif}
There exists a Monte Carlo algorithm for 
\MaxMotif{} that runs in $O(2^kk^2 e\mu)$ time and in polynomial space, 
with the following guarantees: 
(i) the algorithm always returns NO when given a NO-instance as input,
(ii) the algorithm returns YES with probability at least $1/2$ when given a YES-instance as input.
\end{Thm}

\noindent
{\em Remark.} 
We observe that the algorithm in Theorem~\ref{thm:max-motif} runs in linear time in the number of edges $e$ in the host graph $H$. Furthermore, the exponential part $2^k$ of the running time is caused by the sieve \eqref{eq:constrained-sieve}, implying that the algorithm can be executed in parallel on up to $2^k$ processors with essentially linear speedup. A caveat of the algorithm is that it solves only the YES/NO-decision problem, however, it can be extended to extract a solution set $K$ at additional multiplicative cost $k$ to the running time; this extension will be pursued elsewhere. 

\subsection{Weighted edit distance and the closest motif problem}

A natural generalization of the basic graph motif framework is to allow for weighted inexact matches between the ``target'' motif $M$ and a connected induced subgraph. Such variants have been studied in the literature, in particular by Dondi, Fertin, and Vialette~\cite{DFV11-cpm} in the context of either (a) addition of colors to $M$ or (b) substitutions of colors in $M$. We state both problems below as decision problems parameterized by $k$. 

\medskip
\defproblem{%
{\sc Min-Add Graph Motif}~\cite{DFV11-cpm}%
}{%
A connected, undirected host graph $H$, 
a multiset $M$ of colors over a base color set $C$,
a coloring $c:V(H)\rightarrow C$ for the vertices of $H$, 
a positive integer $k$, and a nonnegative integer $d$.
}{%
Is there a subset $K\subseteq V(H)$ of size $k$ such that 
(a) the subgraph induced by $K$ in $H$ 
is connected, and 
(b) it holds that $M\subseteq c(K)$ and $|c(K)\setminus M|\leq d$, 
taking multiplicities into account?
}
\medskip
\defproblem{%
{\sc Min-Substitute Graph Motif}~\cite{DFV11-cpm}%
}{%
A connected, undirected host graph $H$, 
a multiset $M$ of colors over a base color set $C$,
a coloring $c:V(H)\rightarrow C$ for the vertices of $H$, 
a positive integer $k$, and a nonnegative integer $d$.
}{%
Is there a subset $K\subseteq V(H)$ of size $k$ such that 
(a) the subgraph induced by $K$ in $H$ 
is connected, and 
(b) it holds that $M$ can be transformed to $c(K)$ by at most
$d$ substitutions of colors, taking multiplicities into account?
}
\medskip

Koutis~\cite{koutis-ipl} gives an $O^*(2.54^k)$-time algorithm for
{\sc Min-Add Graph Motif} and an $O^*(5.08^k)$-time algorithm 
for {\sc Min-Substitute Graph Motif}. 

Our objective here is to generalize the graph motif framework to weighted edit distance between $M$ and $c(K)$ by introducing a common generalization, the \ClosestMotif{} problem. We then use Lemma~\ref{lem:cost-constrained} to obtain an $O^*(2^k)$-time algorithm for the problem.

We start with some preliminaries to give a precise meaning to ``closest''
via the weighted edit distance.
Let $M$ be a multiset over a base set of colors $C_0$. 
Let us allow to change $M$ by means of three basic operations: 
\begin{enumerate}
\item[(S)] substitute one occurrence of a color $q\in M$ with a color $q'\in C_0$,
\item[(I)] insert one occurrence of a color $q\in C_0$ to $M$, and
\item[(D)] delete one occurrence of a color $q\in M$ from $M$.
\end{enumerate}
Associate with each basic operation (S), (I), (D) an nonnegative integer
{\em cost} $\sigma_{\mathrm{S}}$,
$\sigma_{\mathrm{I}}$,
$\sigma_{\mathrm{D}}$.

For multisets $M$ and $N$ over $C_0$, the {\em cost} 
(or {\em weighted edit distance}) to match $M$ with $N$ 
is the minimum cost of a sequence of basic operations that 
transforms $M$ to $N$, where the cost of the sequence is
the sum of costs of the basic operations in the sequence.

\medskip
\defproblem{%
\ClosestMotif%
}{%
A connected, undirected host graph $H$ with $n$ vertices and $e$ edges, 
a multiset $M$ of colors over a base color set $C_0$,
a coloring $c:V(H)\rightarrow C_0$ for the vertices of $H$, 
nonnegative integer costs $\sigma_{\mathrm{S}}$, $\sigma_{\mathrm{I}}$, $\sigma_{\mathrm{D}}$, 
a threshold cost $\tau$, and 
a positive integer $k$.%
}{%
Is there a subset $K\subseteq V(H)$ of size $k$ such that 
(a) the subgraph induced by $K$ in $H$ 
is connected and 
(b) the cost to transform the multiset $M$ into the multiset $c(K)$
is at most $\tau$?%
}
\medskip

\noindent
{\em Our results.} Our main result in this section is as follows.

\begin{Thm}
\label{thm:closest-motif}
There exists a Monte Carlo algorithm for 
\ClosestMotif{} that runs in $O((2^kk^4+|C_0|k^3)e\mu)$ time and 
in polynomial space, with the following guarantees: 
(i) the algorithm always returns NO when given a NO-instance as input,
(ii) the algorithm returns YES with probability at least $1/2$ when given a YES-instance as input.
\end{Thm}

\noindent
{\em Remark.} Similar remarks apply to Theorem~\ref{thm:closest-motif} as with Theorem~\ref{thm:max-motif}. In particular, the implementation of \eqref{eq:cost-constrained-sieve} with two cost parameters enables essentially linear parallel speedup on up to $2^kk^2$ processors.

\subsection{A lower bound} 

There is some evidence that the exponential 
part $2^k$ in the running time of the algorithms in 
Theorem~\ref{thm:max-motif} and Theorem~\ref{thm:closest-motif} may
be the best possible. Our approach is to proceed by reduction from
the set cover problem.

\medskip
\defproblem{%
\SetCover%
}{%
An integer $t$ and a family of sets $\mathcal{S} = \{S_1,S_2,\ldots,S_m\}$ over {\em the universe} $U=\bigcup_{j=1}^m S_j$ with $n=|U|$.%
}{%
Is there a subfamily of $t$ sets $S_{i_1}, S_{i_2}, \ldots, S_{i_t}$ such that $U=\bigcup_{j=1}^t S_{i_j}$?
}
\medskip

We show that for any $\epsilon>0$ the existence of an $O^*((2-\epsilon)^k)$-time algorithm for \MaxMotif{} implies an $O((2-\epsilon')^n)$-time algorithm for \SetCover{}, for some $\epsilon'>0$. Thus, instead of trying to improve our algorithm one should rather directly attack \SetCover{}, for which all attempts to obtain a $O((2-\epsilon)^n)$-time algorithm have failed, despite extensive effort. Indeed, the nonexistence of such an algorithm is already used as a basis for hardness results~\cite{ccc}. Furthermore, it is conjectured~\cite{ccc} that an $O((2-\epsilon)^n)$-time algorithm for \SetCover{} contradicts the Strong Exponential Time Hypothesis (SETH), which states that if $k$-CNF SAT can be solved in $O^*(c_k^n)$ time, then $\mbox{lim}_{k\rightarrow \infty} c_k=2$. This conjecture is further supported by the fact that the number of solutions to an instance of \SetCover{} cannot be computed in $O((2-\epsilon)^n)$ time for any $\epsilon >0$ unless SETH fails~\cite{ccc}. A yet further consequence of such 
a counting algorithm would be the existence of an $O((2-\epsilon')^n)$-time algorithm to compute the permanent of an $n\times n$ integer matrix~\cite{bjorklund-matchings}.

\begin{Thm}
\label{thm:red}
If \MaxMotif{} can be solved in $O((2-\epsilon)^{k})$ time for some $\epsilon>0$ then \SetCover{} can be solved in $O((2-\epsilon')^{n})$ time, for some $\epsilon'>0$.
Moreover, this holds even for instances of \MaxMotif{} restricted to one of the following two extreme cases:
\begin{enumerate}
 \item each color may occur at most once, or
 \item there are exactly two colors.
\end{enumerate}
\end{Thm}

\subsection{Organization}
Our two main lemmas,
Lemma~\ref{lem:constrained} and Lemma~\ref{lem:cost-constrained},
are proved in \S\ref{sect:proofs}. Theorem~\ref{thm:max-motif}
is proved in \S\ref{sect:max-motif}. Theorem~\ref{thm:closest-motif}
is proved in \S\ref{sect:closest-motif}. Theorem~\ref{thm:red}
is proved in \S\ref{sect:lower-bound}.

\section{Algebraic Tools}
\label{sect:proofs}

This section proves Lemma~\ref{lem:constrained} and 
Lemma~\ref{lem:cost-constrained}. We start with a proof
of Lemma~\ref{lem:basic} that will act as a building block
of both proofs.

\subsection{Proof of Lemma~\ref{lem:basic}.}
It will be convenient to work with a polynomial consisting of a
single monomial, after which it will be easy to extend the analysis 
to an arbitrary polynomial. So suppose that 
\[
P(x_1,x_2,\ldots,x_n)=x_1^{d_1}x_2^{d_2}\cdots x_n^{d_n}
\]
with $d_1+d_2+\ldots+d_n=k$. 
We must show that the expression 
$\sum_{A\subseteq L}P\bigl(z_1^A,z_2^A,\ldots,z_n^A\bigr)$
is not identically zero in characteristic 2
if and only if $d_1,d_2,\ldots,d_n\in\{0,1\}$.

Let us start by simplifying the expression into a more convenient form.
Recalling that $z_i^A=\sum_{j\in A}z_{i,j}$ for $i\in[n]$ and 
expanding the product--sum into a sum--product, we have
\begin{eqnarray}
\notag
\sum_{A\subseteq L}P\bigl(z_1^A,z_2^A,\ldots,z_n^A\bigr)
&=&\sum_{A\subseteq L}\prod_{i=1}^n\biggl(\sum_{j\in A}z_{i,j}\biggr)^{d_i}\\
\notag
&=&\sum_{A\subseteq L}\prod_{i=1}^n\sum_{f_i:[d_i]\rightarrow A}\prod_{\ell=1}^{d_i}z_{i,f_i(\ell)}\\
\label{eq:asum}
&=&\sum_{A\subseteq L}\sum_{f_1:[d_1]\rightarrow A}\sum_{f_2:[d_2]\rightarrow A}\cdots\sum_{f_n:[d_n]\rightarrow A}\prod_{i=1}^n\prod_{\ell=1}^{d_i}z_{i,f_i(\ell)}\,.
\end{eqnarray}
The outer sum in \eqref{eq:asum} is over all subsets $A\subseteq L$ and the
inner sums range over all $n$-tuples $f=(f_1,f_2,\ldots,f_n)$
of functions $f_i:[d_i]\rightarrow A$ with $i\in[n]$. 

Let us fix an arbitrary $n$-tuple $f=(f_1,f_2,\ldots,f_n)$
of functions $f_i:[d_i]\rightarrow L$ with $i\in[n]$. 
Let us define the {\em image} of $f$ by
\[
I(f)=f_1([d_1])\cup f_2([d_2])\cup\cdots\cup f_n([d_n])\,.
\]
Now let us consider the outer sum over subsets $A\subseteq L$ 
in \eqref{eq:asum}. Observe that for a fixed $A\subseteq L$, 
our fixed $n$-tuple $f=(f_1,f_2,\ldots,f_n)$ occurs exactly 
once in the inner sums of \eqref{eq:asum} if and only if $I(f)\subseteq A$.
That is to say, the fixed $f$ occurs exactly once for each $A$ with
$I(f)\subseteq A\subseteq L$. The number of such $A$ is $2^{|L|-|I(f)|}$, 
which is even---and hence cancels in characteristic~2---%
if and only if $I(f)\neq L$. 

Let us say that $f$ is {\em surjective} if $I(f)=L$. 
Since all but surjective $f$ cancel, from \eqref{eq:asum} and
the previous analysis we thus have
\begin{equation}
\label{eq:fsum}
\sum_{A\subseteq L}P\bigl(z_1^A,z_2^A,\ldots,z_n^A\bigr)
=
\sum_{\substack{f=(f_1,f_2,\ldots,f_n)\\f\text{ surjective}}}\prod_{i=1}^n\prod_{\ell=1}^{d_i}z_{i,f_i(\ell)}\,.
\end{equation}
Next we show that \eqref{eq:fsum} is identically zero 
unless $d_1,d_2,\ldots,d_n\in\{0,1\}$. 

So suppose there exists at least one bad index $b\in[n]$ with $d_b\geq 2$. 
Let us fix $b$ to be the minimum such index.
Consider an arbitrary surjective $n$-tuple $f=(f_1,f_2,\ldots,f_n)$.
Since $|L|=k=d_1+d_2+\ldots+d_n$ and $f$ is surjective, we must have that for every $i\in [n]$ the function $f_i$ is bijective, 
in particular thus $f_b(1)\neq f_b(2)$. 

Define the {\em mate} $f'$ of $f$ by setting 
$f_i'=f_i$ for all $i\in [n]\setminus\{b\}$ and
\[
f_b'(\ell)
=\begin{cases}
f_b(2)    & \text{if $\ell=1$;}\\
f_b(1)    & \text{if $\ell=2$;}\\
f_b(\ell) & \text{otherwise.}
\end{cases}
\]
Observe that $f'\neq f$ and that $f''=f$. Thus, the set of
all surjective $f$ partitions into disjoint pairs $\{f,f'\}$
with 
\[
\prod_{i=1}^n\prod_{\ell=1}^{d_i}z_{i,f_i(\ell)}
=
\prod_{i=1}^n\prod_{\ell=1}^{d_i}z_{i,f_i'(\ell)}\,.
\]
Thus, all monomials in \eqref{eq:fsum} have an even coefficient
and hence cancel in characteristic~2 unless $d_1,d_2,\ldots,d_n\in\{0,1\}$.

So suppose that $d_1,d_2,\ldots,d_n\in\{0,1\}$. Since $d_1+d_2+\ldots+d_n=k$,
we have that the set $K=\{i\in[n]:d_i=1\}$ has size $k$. 
Furthermore, associated with each surjective $f$ there is a unique
bijection $g:K\rightarrow L$ defined for all $i\in K$ by $g(i)=f_{i}(1)$.
We thus have
\begin{equation}
\label{eq:gsum}
\sum_{A\subseteq L}P\bigl(z_1^A,z_2^A,\ldots,z_n^A\bigr)
=\sum_{\substack{g:K\rightarrow L\\g\text{ bijective}}}
\prod_{i\in K}z_{i,g(i)}\,.
\end{equation}
In particular, from each monomial $\prod_{i\in K}z_{i,g(i)}$
we can recover both the set $K$ and the bijection $g:K\rightarrow L$,
implying that no cancellation happens in characteristic 2. 
Furthermore, from $K$ we can recover 
$P(x_1,x_2,\ldots,x_n)=\prod_{i\in K}x_i$. 

The lemma now follows by linearity. Indeed, an arbitrary multivariate 
polynomial $P(x_1,x_2,\ldots,x_n)$ is a sum of monomials 
$x_1^{d_1}x_2^{d_2}\cdots x_n^{d_n}$. $\qed$

\subsection{Proof of Lemma~\ref{lem:constrained}.}
We obtain cancellation in characteristic 2 using identical arguments
to the proof of Lemma~\ref{lem:basic}, up to and including adapting
\eqref{eq:gsum} to the setting of Lemma~\ref{lem:constrained}. 
That is, 
\begin{equation}
\label{eq:ggsum}
\sum_{A\subseteq L}P\bigl(u_1^A,u_2^A,\ldots,u_n^A\bigr)
=\sum_{\substack{g:K\rightarrow L\\g\text{ bijective}}}
\prod_{i\in K}u_{i,g(i)}\,.
\end{equation}
We proceed to show that the right-hand side of \eqref{eq:ggsum}
is not identically zero if and only if the multilinear monomial
$\prod_{i\in K} x_i$ is properly colored. 

Let us say that a function $h:K\rightarrow S$ that associates
a shade $h(i)\in S$ to each $i\in K$ is {\em valid} if
it holds that $h(i)\in S_{c(i)}$ for all $i\in K$. 
Observe in particular that an {\em injective} valid $h:K\rightarrow S$
exists if and only if $\prod_{i\in K} x_i$ is properly colored. 

We are now ready to start simplifying the right-hand side of \eqref{eq:ggsum}.
Recalling that $u_{i,j}=\sum_{d\in S_{c(i)}}v_{i,d}w_{d,j}$,
expanding the product--sum into a sum--product, and changing the
order of summation, we have
\begin{eqnarray}
\notag
\sum_{\substack{g:K\rightarrow L\\g\text{ bijective}}}
\prod_{i\in K}u_{i,g(i)}
&=&
\sum_{\substack{g:K\rightarrow L\\g\text{ bijective}}}
\prod_{i\in K}\biggl(\sum_{d\in S_{c(i)}}v_{i,d}w_{d,g(i)}\biggr)\\
\notag
&=&
\sum_{\substack{g:K\rightarrow L\\g\text{ bijective}}}
\sum_{\substack{h:K\rightarrow S\\h\text{ valid}}}
\prod_{i\in K}v_{i,h(i)}w_{h(i),g(i)}\\
\label{eq:hsum}
&=&
\sum_{\substack{h:K\rightarrow S\\h\text{ valid}}}
\sum_{\substack{g:K\rightarrow L\\g\text{ bijective}}}
\prod_{i\in K}v_{i,h(i)}w_{h(i),g(i)}\,.
\end{eqnarray}
The outer sum in \eqref{eq:hsum} ranges over all valid functions 
$h:K\rightarrow S$.

Now, let us fix an arbitrary valid $h:K\rightarrow S$. 
We will show that the inner sum in \eqref{eq:hsum} evaluates
to zero in characteristic 2 unless $h$ is injective. 

So suppose that $h$ is not injective. In particular, there exists
at least one pair $b_1,b_2\in K$ with $h(b_1)=h(b_2)$ and $b_1\neq b_2$. 
Let us fix $(b_1,b_2)$ to be the lexicographically minimum such pair.
Consider an arbitrary bijective $g:K\rightarrow L$. 
Define the {\em mate} $g'$ of $g$ by setting 
\[
g'(i)=
\begin{cases}
g(b_2) & \text{if $i=b_1$;}\\
g(b_1) & \text{if $i=b_2$;}\\
g(i)   & \text{otherwise.}
\end{cases}
\]
Since $g$ is bijective, we have $g'\neq g$ and $g''=g$. 
Thus, the set of all bijections $g:K\rightarrow L$ partitions
into disjoint pairs $\{g,g'\}$ with 
\[
\prod_{i\in K}v_{i,h(i)}w_{h(i),g(i)}
=\prod_{i\in K}v_{i,h(i)}w_{h(i),g'(i)}
\,.
\]
Thus, for each valid $h:K\rightarrow S$ that is not injective,
the monomials in the inner sum in \eqref{eq:hsum} have an even coefficient
and hence vanish in characteristic 2. 

So suppose that $h$ is injective. (Recall that such an $h$ 
exists if and only if $K$ defines a properly colored multilinear
monomial.) Let us study the inner sum in \eqref{eq:hsum}.
Fix an arbitrary bijective $g:K\rightarrow L$ and study 
the inner monomial $\prod_{i\in K}v_{i,h(i)}w_{h(i),g(i)}$. 
From the variables $v_{i,d}$ in the monomial we can reconstruct 
the set $K$ and the mapping $h$. Because $h$ is injective, we can 
reconstruct the mapping $g$ from the variables $w_{d,j}$ 
in the monomial by setting $g(h^{-1}(d))=j$ for each relevant
pair $(d,j)$. Since the three-tuple $(K,h,g)$ can be reconstructed
from the inner monomial, no cancellation happens in characteristic~2.

The lemma follows again by linearity. $\qed$

\subsection{Proof of Lemma~\ref{lem:cost-constrained}.}

Let $\pi:S\rightarrow C$ be the mapping that projects each
shade $d\in S_q$ to its underlying color $\pi(d)=q$. 
Imitating the proof of Lemma~\ref{lem:basic}
and expanding \eqref{eq:cost-u} over $i\in K$ as in 
\eqref{eq:hsum}, we obtain cancellation in characteristic 2, 
except possibly for the monomials 
\begin{equation}
\label{eq:ext-hsum}
\sum_{\substack{h:K\rightarrow S}}
\sum_{\substack{g:K\rightarrow L\\g\text{ bijective}}}
\eta^{\sum_{i\in K}\kappa_i(\pi(h(i)))}
\prod_{i\in K}
v_{i,h(i)}w_{h(i),g(i)}
\,.
\end{equation}
Imitating the proof of Lemma~\ref{lem:constrained}, we obtain
further cancellation in characteristic 2 unless the mapping $h$ 
is injective.

So suppose that $h$ is injective. Observe that we can reconstruct
 the three-tuple $(K,h,g)$ from the corresponding monomial 
in \eqref{eq:ext-hsum} exactly as in the proof of 
Lemma~\ref{lem:constrained}, and thus no further cancellation happens 
in characteristic 2. The degree of $\eta$ is clearly the cost of 
the monomial $\prod_{i\in K}x_i$ in its coloring $c=\pi h$. 
In particular, we have that $\prod_{i\in K}x_i$ is properly colored 
in $c$ since $h$ is injective.

The lemma follows again by linearity. $\qed$

\subsection{Remarks}

It is immediate from the proofs that the polynomial $P(\vec x)$ 
may have additional variables $P(\vec x,\vec y)$ without changing 
the conclusion as regards multilinearity and proper coloring 
of the monomials when restricted to the variables $\vec x$. 
Furthermore, any monomial that has total degree less than 
$k$ in the variables $\vec x$ will cancel. 

We observe that Lemma~\ref{lem:cost-constrained} subsumes Lemma~\ref{lem:constrained}. Indeed, given a coloring $c:[n]\rightarrow C$ we can set the costs for Lemma~\ref{lem:cost-constrained} so that $\kappa_i(q)=0$ if $c(i)=q$ and $\kappa_i(q)=1$ otherwise. Then, $P(\vec x)$ has at least one monomial that is both multilinear and properly colored if and only if $Q(\vec v,\vec w,\eta)$ has at least one monomial whose degree in the variable $\eta$ is $\sigma=0$.

\section{An Algorithm for the Maximum Graph Motif Problem}
\label{sect:max-motif}

This section illustrates the use of Lemma~\ref{lem:constrained} 
in a concrete algorithm design for \MaxMotif{}.
In particular, we proceed to give a proof of Theorem~\ref{thm:max-motif}.

Consider an instance $(H,M,C,c,k)$ of \MaxMotif{}. Let us write
$m(q)$ for the number of occurrences of color $q\in C$ in 
the multiset $M$. Also recall that we assume that the host graph
$H$ is connected with $n$ vertices and $e$ edges; in particular, 
$e\geq n-1$. By preprocessing we may assume that $m(q)\leq k$ for 
each $q\in C$.

Our first objective is to arrive at a generating polynomial 
$P_k(\vec x,\vec y)$ that we can use with Lemma~\ref{lem:constrained}. 
There are two key aspects to this quest: (i) the multilinear monomials 
need to reflect the connected vertex sets of size $k$ in $H$, and 
(ii) we must have a fast algorithm for evaluating the polynomial at 
specific points. 

\subsection{Branching walks}

The concept of branching walks was first introduced by 
Nederlof~\cite{nederlof-steiner} to sieve for Steiner trees,
followed by Guillemot and Sikora~\cite{gs10} who 
observed that branching walks can also be employed to span 
connected vertex sets of size $k$ in the host graph $H$. 
Our approach here is to capitalize on this observation and 
span connected sets via branching walks.

Let us write $V=V(H)=\{1,2,\ldots,n\}$ for the vertex set 
and $E=E(H)$ for the edge set of the host graph $H$.
A mapping $\varphi:V(T)\rightarrow V(H)$ is a {\em homomorphism} 
from a graph $T$ to the host $H$ if for all $\{a,b\}\in E(T)$ it holds 
that $\{\varphi(a),\varphi(b)\}\in E(H)$. We adopt the convention of 
calling the elements of $V(T)$ {\em nodes} and the elements 
of $V(H)$ {\em vertices}. 

A {\em branching walk} in $H$ is a pair $W=(T,\varphi)$ where $T$ is 
an ordered rooted tree with node set $V(T)=\{1,2,\ldots,|V(T)|\}$ such 
that every node $a\in V(T)$ coincides with its rank in the preorder 
traversal of $T$, and $\varphi:V(T)\rightarrow V(H)$ is a homomorphism 
from $T$ to $H$. 

Let $W=(T,\varphi)$ be a branching walk in $H$.
The walk {\em starts} from the vertex $\varphi(1)$ in $H$.
The walk {\em spans} the vertices $\varphi(V(T))$ in $H$.
The {\em size} of the walk is $|V(T)|$. 
The walk is {\em simple} if $\varphi$ is injective.
Finally, the walk is {\em properly ordered} if any two sibling 
nodes $a<b$ in $T$ satisfy $\varphi(a)<\varphi(b)$ in $H$.

\subsection{A generating polynomial for branching walks}
\label{sect:genf}

We now define a generating polynomial for properly ordered branching walks 
of size $k$ in $H$. Introduce a variable $x_u$ for each 
vertex $u\in V(H)$ and two variables $y_{(u,v)}$ and $y_{(v,u)}$ 
for each edge $\{u,v\}\in E(H)$.

Let $W=(T,\varphi)$ be a properly ordered branching walk that
starts from $s\in V(H)$ and has size $k$.
Associate with $W$ the {\em monomial fingerprint}
\[
F(W,\vec x,\vec y)=
\prod_{\substack{\{a,b\}\in E(T)\\a<b}}
y_{(\varphi(a),\varphi(b))}x_{\varphi(b)}\,,
\]
where the product is taken over all edges $\{a,b\}\in E(T)$.

Define the generating polynomial $P_{k,s}(\vec x,\vec y)$ as the sum
of the monomial fingerprints of the properly ordered 
branching walks that start from $s$ and have size $k$.
Let $P_k(\vec x,\vec y)=\sum_{s\in V(H)} x_{s} P_{k,s}(\vec x,\vec y)$.
Observe that all monomial in $P_k(\vec x,\vec y)$ have total
degree $2k-1$.

\begin{Lem}
\label{lem:fingerprint}
A monomial in $P_k(\vec x,\vec y)$ is multilinear in the variables 
$\vec x$ if and only if it originates from a monomial fingerprint of 
a simple branching walk. Moreover, such a simple branching 
walk can be reconstructed from its monomial fingerprint.
\end{Lem}

\begin{proof}
For the first claim it suffices to consider an arbitrary 
monomial of $P_k(\vec x,\vec y)$ and observe that 
the degree of the variable $x_u$ indicates how many times $u\in V(H)$ 
occurs in the image of $\varphi$. In particular, $\varphi$ is injective 
if and only if the monomial is multilinear in the variables $\vec x$.

For the second claim, 
let $W=(T,\varphi)$ be a simple and properly ordered branching
walk that starts from $s$. We must reconstruct $W$ from its monomial 
fingerprint that has been multiplied by $x_s$.
Since $\varphi$ is injective, we can immediately reconstruct 
(up to labels of the vertices) the rooted tree structure of $T$ because 
the degrees 
of the variables $y_{(u,v)}$ in the monomial (if any) reveal both 
the edges and the orientation of each edge in $T$. Since $W$ is 
properly ordered, we can reconstruct (up to labels of the vertices) 
the ordering of $T$. Finally, we can reconstruct the vertex labels 
of $T$ by carrying out a preorder traversal of $T$.
\end{proof}

An immediate corollary of Lemma~\ref{lem:fingerprint} is that
$(H,M,C,c,k)$ is a YES-instance of \MaxMotif{} if and only
if $P_k(\vec x,\vec y)$ has a monomial that is both properly colored and 
multilinear in the sense of Lemma~\ref{lem:constrained}. 
Indeed, a multilinear monomial corresponds to a simple branching walk,
which by definition spans a connected set of vertices. Conversely, every
connected set of vertices admits at least one simple branching walk.
Thus, to complete the proof of Theorem~\ref{thm:max-motif} it remains 
to derive a fast way to evaluate the polynomial $P_k(\vec x,\vec y)$ and 
then apply Lemma~\ref{lem:constrained} to obtain an algorihtm design.

\subsection{Evaluating the generating polynomial}

This section develops a dynamic programming recurrence to 
evaluate the polynomial $P_{k}(\vec x,\vec y)$ at a given 
assignment of values to the variables $\vec x,\vec y$. 

For a vertex $u\in V(H)$, denote the ordered sequence of neighbors 
of $u$ in $H$ by $u_1<u_2<\cdots<u_{\deg_H(u)}$.

For each $u\in V(H)$, $1\leq i\leq \deg_H(u)+1$, 
and $0\leq\ell\leq k$, denote by $\mathcal{W}(\ell,u,i)$ the set of 
properly ordered branching walks $W=(T,\varphi)$ such that 
(i) the size of $W$ is $\ell$,
(ii) $W$ starts from $u$, and
(iii) for any child node $a$ of $1$ in $T$ it 
holds that $\varphi(a)=u_j$ implies $j\ge i$.
Define the associated generating polynomial over the variables 
$\vec x,\vec y$ by
\[
P_{\ell,u,i}(\vec x,\vec y)=
\sum_{(T,\varphi)\in\mathcal{W}(u,1,\ell)}
\prod_{\substack{\{a,b\}\in E(T)\\a<b}}
y_{(\varphi(a),\varphi(b))}x_{\varphi(b)}\,.
\]
It is immediate from the definition that $P_{\ell,u}(\vec x,\vec y)=P_{\ell,u,1}(\vec x,\vec y)$.

The functions $P_{\ell,u,i}(\vec x,\vec y)$ admit the following recurrence.
The base case occurs for $\ell=1$ or $i=\deg_H(u)+1$, in which
case we have
\begin{equation}
\label{eq:dp1}
P_{\ell,u,i}(\vec x,\vec y)=
\begin{cases}
1 & \text{if $\ell=1$,}\\
0 & \text{otherwise}.
\end{cases}
\end{equation}
For $2\leq \ell \leq k$ and $1\leq i\leq \deg_H(u)$, we have
\begin{eqnarray}
\nonumber
 P_{\ell,u,i}(\vec x,\vec y) & = & P_{\ell,u,i+1}(\vec x,\vec y) + \\
 \label{eq:dp2}
  && y_{(u,u_i)}x_{u_i}\sum_{\substack{\ell_1+\ell_2=\ell\\\ell_1,\ell_2\ge 1}} P_{\ell_1,u,i+1}(\vec x,\vec y)\cdot P_{\ell_2,u_i,1}(\vec x,\vec y)\,.
\end{eqnarray}
To see that the recurrence is correct, observe that 
the two lines above in \eqref{eq:dp2} correspond to a partitioning
of the properly ordered branching walks in $\mathcal{W}(\ell,u,i)$ 
into two disjoint classes where either 
(i) there is no child node $a$ of $1$ in $T$ such that $h(a)=u_i$ 
or 
(ii) there is a unique such child. 
(At most one such child may exist because the branching walk 
is properly ordered.)

Thus, we can evaluate the polynomial $P_{k}(\vec x,\vec y)$ 
via \eqref{eq:dp1}, \eqref{eq:dp2}, and 
\begin{equation}
\label{eq:dp-top}
P_{k}(\vec x,\vec y)=\sum_{u\in V(H)}x_uP_{k,u,1}(\vec x,\vec y)\,.
\end{equation}

\subsection{The algorithm}

We are now ready to describe the algorithm for Theorem~\ref{thm:max-motif}.
Assume an instance $(H,M,C,c,k)$ of the \MaxMotif{} has been given as input. 

Let $b=\lceil \log_2 6k\rceil$ and consider the finite field 
$\mathbb{F}_{2^b}$ of order $2^b$.
Introduce variables $v_{i,d}$ and $w_{d,j}$ as in the setup
of Lemma~\ref{lem:constrained}. Assign a value from $\mathbb{F}_{2^b}$
uniformly and independently at random to each of these variables. 
Similarly, as in the setup of \S\ref{sect:genf}, introduce two variables
$y_{(r,s)}$ and $y_{(s,r)}$  to each edge $\{r,s\}\in E(H)$ and assign 
a value to each variable uniformly and independently at random from 
$\mathbb{F}_{2^b}$.
We thus have three vectors of values in $\mathbb{F}_{2^b}$, namely 
$\vec v$, $\vec w$, and $\vec y$.

Using the recurrence given by \eqref{eq:dp1}, \eqref{eq:dp2}, 
and \eqref{eq:dp-top} for each $A\subseteq L$ in turn, 
compute the value 
\begin{equation}
\label{eq:q-motif}
Q(\vec v,\vec w,\vec y)=
\sum_{A\subseteq L}P_k\bigl(\vec u^A(\vec v,\vec w),\vec y\bigr)\,,
\end{equation}
where the values $\vec u^A(\vec v,\vec w)=(u_1^A(\vec v,\vec w),u_2^A(\vec v,\vec w),\ldots,u_n^A(\vec v,\vec w))$ are determined from the set $A$ and the 
values $\vec v$ and $\vec w$ as in Lemma~\ref{lem:constrained}. 
If $Q(\vec v,\vec w,\vec y)$ is nonzero in $\mathbb{F}_{2^b}$, 
output YES; otherwise output NO. This completes the description of 
the algorithm.

\subsection{Running time}
\label{sect:max-motif-runtime}

To analyse the running time of the algorithm, observe that we can assume 
that $m(q)\leq k$. Thus, computing the values $\vec u^A(\vec v,\vec w)$ for 
a fixed $A\subseteq L$ takes $O(k^2n)$ arithmetic operations in 
$\mathbb{F}_{2^b}$, and each such operation can be implemented to 
run in time $\mu=O(b\log b\log\log b)$~\cite{algebraic-complexity-theory}.
Furthermore, each evaluation of \eqref{eq:dp1}, \eqref{eq:dp2}, 
and \eqref{eq:dp-top} for a fixed $A$ takes $O(k^2e)$ 
arithmetic operations in $\mathbb{F}_{2^b}$. Hence, recalling that
$e\geq n-1$, the total running time of the algorithm is $O(2^kk^2e\mu)$.

\subsection{Correctness}

To establish the desired properties of the algorithm, observe that 
from \S\ref{sect:genf} and Lemma~\ref{lem:constrained} it follows 
that \eqref{eq:q-motif}%
---viewed as a polynomial in the variables $\vec v$, $\vec w$, 
and $\vec y$---is not identically zero if and only if $(H,M,C,c,k)$ 
is a YES-instance of \MaxMotif{}. 
Thus, if $(H,M,C,c,k)$ is a NO-instance, 
then \eqref{eq:q-motif} evaluates to zero and the algorithm gives 
a NO output. Furthermore, if $(H,M,C,c,k)$ is a YES-instance, 
then \eqref{eq:q-motif} is an evaluation of a nonzero multivariate 
polynomial of total degree $3k-1$ at a point $(\vec v,\vec w,\vec y)$ 
selected uniformly at random. 
Recalling that $2^b\geq 6k$, the following lemma thus implies 
that the value $Q(\vec v,\vec w,\vec y)$ is nonzero (and hence 
the algorithm outputs YES) with probability at least 1/2.

\begin{Lem}[\cite{DeMilloLipton1978,schwartz,zippel}]
\label{lem:dmlsz}
A nonzero polynomial $P(z_1,z_2,\ldots,z_\ell)$ of 
total degree $d$ with coefficients in the finite field $\mathbb{F}_q$
has at most $dq^{\ell-1}$ roots in $\mathbb{F}_q^\ell$. 
\end{Lem}

This completes the proof of Theorem~\ref{thm:max-motif}.$\qed$

\subsection{Minor variants and extensions}

The basic framework presented above immediately allows for some 
minor variants and extensions, such as seeking an exact match instead
of the maximum match by setting $|M|=k$. Similarly, one may extend
from a fixed coloring $c:V(H)\rightarrow C$ into a {\em list coloring} 
version where each vertex $i\in V(H)$ gets associated a list 
$C(i)\subseteq C$ of valid colors instead of a single color $c(i)$, and 
the motif $M$ may match against any one of the colors in the list. 
This variant can be implemented by simply changing the inner sum in 
Lemma~\ref{lem:constrained} to $u_{i,j}=\sum_{d\in \cup_{q\in C(i)}S_q}v_{i,d}w_{d,j}$. That is, we sum over the shades of all the colors $q$ in $C(i)$.

\section{An Algorithm for the Closest Graph Motif Problem}
\label{sect:closest-motif}

This section gives a proof of Theorem~\ref{thm:closest-motif} using
Lemma~\ref{lem:cost-constrained} and the generating function 
developed in \S\ref{sect:genf}. 

Consider an instance $(H,M,C_0,c,\sigma_{\mathrm{S}},\sigma_{\mathrm{I}},\sigma_{\mathrm{D}},\tau,k)$ of \ClosestMotif{} with $V(H)=\{1,2,\ldots,n\}$. Let us again write $m(q)$ for the number of occurrences of color $q\in C_0$ in the multiset $M$. We may assume that $m(q)\leq k$. Furthermore, since $H$ is connected, the number of vertices $n$ and the number of edges $e$ satisfy $e\geq n-1$.

The key step in arriving at Theorem~\ref{thm:closest-motif} is to 
transport weighted edit distance into the setting of 
Lemma~\ref{lem:cost-constrained}.

\subsection{Optimum edit sequences}

It will be convenient to have available the following lemma
that characterizes the structure of a sequence of operations
that realizes the minimum cost to transform a multiset $M$ to 
the multiset $N$, where both multisets are over $C_0$.

Let $k=|N|$. Consider an arbitrary sequence of basic operations 
that transforms $M$ to $N$. As the sequence is executed, 
each original element of $M$ gets assigned into one of three classes.
First, there are $k_{\mathrm{U}}$ elements in $M$ that remain
untouched (and hence in $N$) when the execution terminates.
Second, there are $k_{\mathrm{S}}$ elements in $M$ that undergo
at least one substitution---which we may view as ``recoloring'' 
of the element---and remain in $N$ when the execution terminates.
Third, the remaining $|M|-k_{\mathrm{U}}-k_{\mathrm{S}}$ elements 
of $M$ get deleted during execution. Thus, at least 
$k-k_{\mathrm{U}}-k_{\mathrm{S}}$ insertions must occur in 
the sequence. Let us call the values $k_{\mathrm{U}}$ and $k_{\mathrm{S}}$
the {\em parameters} of the sequence. 

\begin{Lem}
\label{lem:optimum-parameterized-sequence}
Let there exist at least one sequence with parameters 
$k_{\mathrm{U}}$ and $k_{\mathrm{S}}$ that transforms 
$M$ into $N$. Then, the cost of this sequence is at least
\begin{equation}
\label{eq:optimum-parameterized-sequence}
\sigma_{\mathrm{S}}k_{\mathrm{S}}
+\sigma_{\mathrm{D}}\bigl(|M|-k_{\mathrm{U}}-k_{\mathrm{S}}\bigr)
+\sigma_{\mathrm{I}}\bigl(k-k_{\mathrm{U}}-k_{\mathrm{S}}\bigr),
\end{equation}
with equality for at least one sequence that transforms $M$ into $N$. 
\end{Lem}

\begin{proof}
The inequality is immediate from the preceding analysis;
the sequence that meets equality (i) does nothing
for the $k_{\mathrm{U}}$ untouched original elements, 
(ii) substitutes the correct final color with one substitution 
for each of the $k_{\mathrm{S}}$ originals, (iii) deletes 
each of the $|M|-k_{\mathrm{U}}-k_{\mathrm{S}}$ remaining originals, and
(iv) finally inserts $k-k_{\mathrm{U}}-k_{\mathrm{S}}$ new elements 
to match with $N$. 
\end{proof}

Lemma~\ref{lem:optimum-parameterized-sequence} reveals a useful
symmetry between insertions and deletions in an optimum sequence;
that is, if we let $k_{\mathrm{ID}}=k-k_{\mathrm{U}}-k_{\mathrm{S}}$,
then \eqref{eq:optimum-parameterized-sequence} is equal to
\begin{equation}
\label{eq:total-parameterized-cost}
\sigma_{\mathrm{S}}k_{\mathrm{S}}
+\bigl(\sigma_{\mathrm{I}}+\sigma_{\mathrm{D}}\bigr)k_{\mathrm{ID}}
+\sigma_{\mathrm{D}}\bigl(|M|-k\bigr)\,.
\end{equation}
Thus it suffices to optimize over $k$-multisets of colors while
tracking the parameters $k_{\mathrm{S}}$ and $k_{\mathrm{ID}}$ 
to arrive at the optimum. This strategy will be employed in our 
algorithm. 

\subsection{The algorithm}
Assume an instance $(H,M,C_0,c,\sigma_{\mathrm{S}},\sigma_{\mathrm{I}},\sigma_{\mathrm{D}},\tau,k)$ of \ClosestMotif{} has been given as input. 

Let us first set up the application of Lemma~\ref{lem:cost-constrained}. 
Introduce a new color ``$*$'' and let $C=C_0\cup\{*\}$ with $m(*)=k$. 
As already highlighted in the remarks to Lemma~\ref{lem:cost-constrained},
instead of one indeterminate $\eta$, we will work with two indeterminates 
$\eta_{\mathrm{S}}$ and $\eta_{\mathrm{ID}}$ in 
Lemma~\ref{lem:cost-constrained} to simultaneously track the 
$\mathrm{S}$-cost and the $\mathrm{ID}$-cost.
For $i\in [n]$ and $q\in C$, define the cost functions
\begin{equation}
\label{eq:s-cost}
\kappa_i^{\mathrm{S}}(q)=
\begin{cases}
0 & \text{if $q=c(i)$;}\\
1 & \text{if $q\neq c(i)$ and $q\in C_0$;}\\
0 & \text{if $q=*$}
\end{cases}
\end{equation}
and
\begin{equation}
\label{eq:id-cost}
\kappa_i^{\mathrm{ID}}(q)=
\begin{cases}
0 & \text{if $q=c(i)$;}\\
0 & \text{if $q\neq c(i)$ and $q\in C_0$;}\\
1 & \text{if $q=*$.}
\end{cases}
\end{equation}
The intuition underlying \eqref{eq:s-cost} and \eqref{eq:id-cost} is as
follows. Coloring a vertex $i$ with color $q\notin\{c(i),*\}$ corresponds
to substitution of a copy of $q$ in $M$ by a copy of $c(i)$. 
Coloring $i$ with color ``$*$'' corresponds to inserting a copy of
$c(i)$ to $M$. 

The algorithm now proceeds as follows.
Let $b=\lceil \log_2 6k\rceil$ and consider the finite field 
$\mathbb{F}_{2^b}$ of order $2^b$.
Introduce variables $v_{i,d}$ and $w_{d,j}$ as in the setup
of Lemma~\ref{lem:cost-constrained}. Assign a value from $\mathbb{F}_{2^b}$
uniformly and independently at random to each of these variables. 
Similarly, as in the setup of \S\ref{sect:genf}, introduce two variables
$y_{(r,s)}$ and $y_{(s,r)}$  to each edge $\{r,s\}\in E(H)$ and assign 
a value to each variable uniformly and independently at random from 
$\mathbb{F}_{2^b}$.
We thus have three vectors of values in $\mathbb{F}_{2^b}$, namely 
$\vec v$, $\vec w$, and $\vec y$.

The main part of the algorithm consists of two outer loops that cycle 
through $k+1$ distinct values 
in $\mathbb{F}_{2^b}$ to each of the variables $\eta_{\mathrm{S}}$ and 
$\eta_{\mathrm{ID}}$. For each pair of values
$(\eta_{\mathrm{S}},\eta_{\mathrm{ID}})$ in $\mathbb{F}_{2^b}$,
we use the recurrence given by \eqref{eq:dp1}, \eqref{eq:dp2}, 
and \eqref{eq:dp-top} for each $A\subseteq L$ in turn, 
and compute the value
\begin{equation}
\label{eq:q-closest-motif}
Q(\vec v,\vec w,\vec y,\eta_{\mathrm{S}},\eta_{\mathrm{ID}})=
\sum_{A\subseteq L}P_k\bigl(\vec u^A(\vec v,\vec w,\eta_{\mathrm{S}},\eta_{\mathrm{ID}}),\vec y\bigr)\,,
\end{equation}
where the values 
\[
\vec u^A(\vec v,\vec w,\eta_{\mathrm{S}},\eta_{\mathrm{ID}})=
(u_1^A(\vec v,\vec w,\eta_{\mathrm{S}},\eta_{\mathrm{ID}}),u_2^A(\vec v,\vec w,\eta_{\mathrm{S}},\eta_{\mathrm{ID}}),\ldots,u_n^A(\vec v,\vec w,\eta_{\mathrm{S}},\eta_{\mathrm{ID}})) 
\]
are determined from the set $A$ and the values $\vec v$ and $\vec w$ as 
in Lemma~\ref{lem:cost-constrained}, but with \eqref{eq:cost-u} replaced by
\begin{equation}
\label{eq:bivar-cost-u}
u_{i,j}=\sum_{q\in C}\eta_{\mathrm{S}}^{\kappa_i^{\mathrm{S}}(q)}\eta_{\mathrm{ID}}^{\kappa_i^{\mathrm{ID}}(q)}
\sum_{d\in S_q}v_{i,d}w_{d,j}
\end{equation}
for all $i\in[n]$, $j\in L$, and $A\subseteq L$.
When the main part terminates, we have available 
$(k+1)^2$ evaluations of \eqref{eq:q-closest-motif} 
at points $(\eta_{\mathrm{S}},\eta_{\mathrm{ID}})$.

By Lagrange interpolation, we recover \eqref{eq:q-closest-motif}
as a bivariate polynomial of total degree at most $k$ 
in the indeterminates $\eta_{\mathrm{S}}$ and $\eta_{\mathrm{ID}}$.
If this bivariate polynomial has at least one monomial 
$\eta_{\mathrm{S}}^{k_{\mathrm{S}}}
 \eta_{\mathrm{ID}}^{k_{\mathrm{ID}}}$
such that the degrees 
$k_{\mathrm{S}}$ and $k_{\mathrm{ID}}$ satisfy
\begin{equation}
\label{eq:cost-comparison}
\sigma_{\mathrm{S}}k_{\mathrm{S}}
+\bigl(\sigma_{\mathrm{I}}+\sigma_{\mathrm{D}}\bigr)k_{\mathrm{ID}}
+\sigma_{\mathrm{D}}\bigl(|M|-k\bigr)\leq \tau\,,
\end{equation}
then the algorithm outputs YES; otherwise the algorithm outputs NO. 
This completes the description of the algorithm.

\subsection{Running time}
The analysis is essentially similar to \S\ref{sect:max-motif-runtime}, 
with two differences. First, the outer loop in the main part introduces
a multiplicative factor $k^2$ compared with \S\ref{sect:max-motif-runtime}.
Second, the implementation of \eqref{eq:bivar-cost-u} requires us to 
sum over all the shades originating from $M$ and the $k$ shades
of the color ``$*$''. This can be done efficiently by precomputing 
the inner sums $\sum_{d\in S_q}v_{i,d}w_{d,j}$ for each color 
$q\in C$, index $i\in [n]$, and label $j\in L$, which takes 
$O\bigl((|M|+k)kn\mu\bigr)$ time outside the main loops. In 
the outer loop of the main part it thus suffices to compute only 
the outer sum in \eqref{eq:bivar-cost-u} for each choice of 
$(\eta_{\mathrm{S}},\eta_{\mathrm{ID}})$, which leads to 
$O\bigl(|C_0|kn\mu\bigr)$ time for each iteration of the outer
loop. In the inner loop over $A\subseteq L$, it takes 
$O(kn\mu)$ time to prepare the vector 
$\vec u^A(\vec v,\vec w,\eta_{\mathrm{S}},\eta_{\mathrm{ID}})$.
Compared with \S\ref{sect:max-motif-runtime},
this gives a further contributing factor of $|C_0|k$ outside the
inner loop.
(The running time cost of the final interpolation step and the checking of 
the at most $k^2$ monomials of the bivariate polynomial 
$Q(\vec v,\vec w,\vec y,\eta_{\mathrm{S}},\eta_{\mathrm{ID}})$ 
with respect to \eqref{eq:cost-comparison} is assumed
to be subsumed by the running time bound.)

\subsection{Correctness}

We start by observing that \eqref{eq:s-cost} and \eqref{eq:id-cost} imply 
that \eqref{eq:q-closest-motif} has total degree at most $k$ in
the variables $\eta_{\mathrm{S}}$ and $\eta_{\mathrm{ID}}$,
thus implying that Lagrange interpolation will correctly
recover the polynomial in $\eta_{\mathrm{S}}$ and $\eta_{\mathrm{ID}}$
from the $(k+1)^2$ evaluations computed in the main loop.

Let us say that \eqref{eq:q-closest-motif}%
---viewed as a polynomial in all the variables $\vec v$, $\vec w$, 
$\vec y$, $\eta_{\mathrm{S}}$, $\eta_{\mathrm{ID}}$---is
{\em witnessing} if there exists at least one monomial 
whose degrees $k_{\mathrm{S}}$ and $k_{\mathrm{ID}}$ satisfy
\eqref{eq:cost-comparison}. 

\begin{Lem}
The polynomial \eqref{eq:q-closest-motif} is witnessing
if and only if the given input is a YES-instance of \ClosestMotif{}. 
\end{Lem}

\begin{proof}
In the ``only if'' direction, consider a monomial of 
\eqref{eq:q-closest-motif} whose degrees 
$k_{\mathrm{S}}$ and $k_{\mathrm{ID}}$ satisfy \eqref{eq:cost-comparison}. 
From Lemma~\ref{lem:cost-constrained}
we have that the polynomial $P_k(\vec x,\vec y)$ has
at least one monomial that is both multilinear in $\vec x$
and admits a proper coloring with $\mathrm{S}$-cost 
$k_{\mathrm{S}}$ and $\mathrm{ID}$-cost $k_{\mathrm{ID}}$.
From \S\ref{sect:genf} it follows that this monomial of 
$P_k(\vec x,\vec y)$ corresponds to a simple branching walk
in $H$ and thus identifies a connected set $K\subseteq V(H)$ 
of vertices in $H$. Furthermore, the existence
of a proper coloring of the monomial implies by 
\eqref{eq:s-cost}, \eqref{eq:id-cost}, and 
Lemma~\ref{lem:optimum-parameterized-sequence} 
that there exists a sequence of basic operations that
transforms the multiset $M$ to the multiset $c(K)$
with total cost \eqref{eq:total-parameterized-cost}.
In particular, since $k_{\mathrm{S}}$ and $k_{\mathrm{ID}}$ satisfy
\eqref{eq:cost-comparison}, we have that 
$(H,M,C_0,c,\sigma_{\mathrm{S}},\sigma_{\mathrm{I}},\sigma_{\mathrm{D}},\tau,k)$
is a YES-instance of \ClosestMotif{}. 

In the ``if'' direction, let 
$(H,M,C_0,c,\sigma_{\mathrm{S}},\sigma_{\mathrm{I}},\sigma_{\mathrm{D}},\tau,k)$
be a YES-instance of \ClosestMotif{}. 
Let $K\subseteq V(H)$ be a solution set and consider an associated 
sequence $\Delta$ of basic operations that transforms $M$ to $c(K)$ with cost 
at most $\tau$. We may without loss of generality assume that the cost of the 
sequence $\Delta$ satisfies equality 
in Lemma~\ref{lem:optimum-parameterized-sequence}. 
In particular, from \eqref{eq:total-parameterized-cost} we thus observe
that the parameters $k_{\mathrm{S}}$ and $k_{\mathrm{ID}}$ 
of the sequence $\Delta$ thus satisfy \eqref{eq:cost-comparison}. 
Consider a simple branching walk of size $k$ in $H$ that spans 
the vertices in $K$. From \S\ref{sect:genf} we observe that
there is a corresponding multilinear monomial in $P_k(\vec x,\vec y)$. 
Next observe that we can properly color this monomial 
in the sense of Lemma~\ref{lem:cost-constrained} by 
(i) 
assigning the color $*$ to each of the $k_{\mathrm{ID}}$ 
values $i\in K$ that correspond to elements inserted in $\Delta$,
(ii) 
assigning the substituted color to each of the 
$k_{\mathrm{S}}$ values $i\in K$ that correspond to elements 
of $M$ receiving substitutions in $\Delta$, 
and (iii)
assigning the color $c(i)$ to each
of the remaining $k-k_{\mathrm{S}}-k_{\mathrm{ID}}$ values
$i\in K$ that correspond to elements of $M$ 
that are not touched by $\Delta$.
Furthermore, by \eqref{eq:s-cost} and \eqref{eq:id-cost}, 
this proper coloring has 
$\mathrm{S}$-cost $k_{\mathrm{S}}$ and
$\mathrm{ID}$-cost $k_{\mathrm{ID}}$.
From Lemma~\ref{lem:cost-constrained} we thus have 
that \eqref{eq:q-closest-motif}%
---viewed as a polynomial in the variables $\vec v$, $\vec w$, 
$\vec y$, $\eta_{\mathrm{S}}$, $\eta_{\mathrm{ID}}$---has
at least one monomial whose degrees 
$k_{\mathrm{S}}$ and $k_{\mathrm{ID}}$ satisfy \eqref{eq:cost-comparison}.
\end{proof}

Let us now study the operation of the algorithm in more detail.
We have that the given input is a NO-instance if and only if 
\eqref{eq:q-closest-motif} is not witnessing. 
Thus,
given a NO-instance as input, the algorithm always gives a NO output.

So suppose that the given input is a YES-instance.
Since \eqref{eq:q-closest-motif} is witnessing, there exist
degrees $k_{\mathrm{S}}$ and $k_{\mathrm{ID}}$
that are present in a monomial of \eqref{eq:q-closest-motif}
such that \eqref{eq:cost-comparison} holds. 
In particular, coefficient of the monomial 
$\eta_{\mathrm{S}}^{k_{\mathrm{S}}}\eta_{\mathrm{ID}}^{k_{\mathrm{ID}}}$
computed by the algorithm is an evaluation of a nonzero multivariate 
polynomial of total degree $3k-1$ at a point $(\vec v,\vec w,\vec y)$ 
selected uniformly at random. 
Recalling that $2^b\geq 6k$, Lemma~\ref{lem:dmlsz} thus implies 
that the coefficient is nonzero (and hence the algorithm outputs YES) 
with probability at least 1/2.
This completes the proof of Theorem~\ref{thm:closest-motif}.$\qed$

\section{A Lower Bound Reduction from Set Cover}

\label{sect:lower-bound}

We base our proof of Theorem~\ref{thm:red} on the following
theorem, which can be extracted from the proof of Theorem 4.4 
in a recent paper by Cygan {\em et al.}~\cite{ccc}.

\begin{Thm}[\cite{ccc}]
\label{thm:setcover}
If \SetCover{} can be solved in $O((2-\epsilon)^{n+t})$ time for some $\epsilon>0$ then it can also be solved in $O((2-\epsilon')^{n})$ time, for some $\epsilon'>0$.
\end{Thm}

\subsection{Proof of Theorem~\ref{thm:red}}

Let $(\mathcal{S},t)$ be an instance of \SetCover{}.
We are going to show a polynomial-time reduction to \MaxMotif{} so that in the resulting instance $(H,C,m,c,k)$ we have $\sum_{q\in C}m(q)=k=n+t+1$. Combined with Theorem~\ref{thm:setcover}, this reduction will immediately establish our claim.

The graph $H$ is defined as follows. The vertex set consists of the universe $U$, $t$ copies of the family $\mathcal{S}$, and a special vertex $r$, that is, $V(H) = U \cup \{s^j_i:i=1,2,\ldots,m,\ j=1,2,\ldots,t\} \cup \{r\}$. The edge set is $E(H)=\{\{a,s^j_i\}:a \in S_i\} \cup \{\{r,s^j_i\}:i=1,2,\ldots,m,\ j=1,2,\ldots,t\}$. Let $k=n+t+1$.

To establish part (1), let $C=\{1,2,\ldots,n+t+1\}$ with $m(q)=1$ for each $q\in C$. Furthermore, assign the colors to vertices so that $c(s^j_i)=j$ for every $i=1,2,\ldots,m,\ j=1,2,\ldots,t$ and $c(r)=t+1$. Finally, assign the $n$ colors $t+2,t+3,\ldots,n+t+1$ bijectively to the vertices in $U$.

We show that $(\mathcal{S},t)$ is a YES-instance if and only if $(H,C,m,c,k)$ is a YES-instance. To establish the ``only if'' direction, suppose that $S_{i_1},S_{i_2},\ldots,S_{i_t}$ is a solution of $(\mathcal{S},t)$. Then let $K=\{r\}\cup U \cup \{s^{j}_{i_j}:j=1,2,\ldots,t\}$. It is clear that $c(K)=C$ and that $H[\{r\} \cup \{s^{j}_{i_j}:j=1,2,\ldots,t\}]$ is connected. Since for every $a\in U$ there is $j=1,2,\ldots,t$ such that $a\in S_{i_j}$, so $\{a,s^j_{i_j}\} \in E(G[K])$. It follows that $G[K]$ is connected, and hence $K$ is a solution of $(H,C,m,c,K)$. To establish the ``if'' direction, suppose that $K$ is a solution of $(H,C,m,c,k)$. Then for every $j=1,2,\ldots,t$ there is exactly one $i_j\in\{1,2,\ldots,m\}$ such that $s^j_{i_j} \in K$, since $c(K)=C$. Moreover, since $G[K]$ is connected we observe that for every $a\in U$ there is a $j=1,2,\ldots,t$ such that $\{a,s^j_{i_j}\}\in E(G[K])$. But then $a\in S_{i_j}$ and it follows that $S_{i_1},S_{i_2},\ldots,S_{i_t}$ is a solution of $(\mathcal{S}
,t)$.

To establish part (2), let $C=\{1,2\}$ with $m(1)=n+1$ and $m(2)=t$. Set $c(r)=1$ and $c(a)=1$ for every $a\in U$. All the remaining vertices are colored with 2. The proof of equivalence is similar to part (1) and is left to the reader.
$\qed$

\section*{Acknowledgments}

A preliminary conference abstract of this work has appeared as \cite{our-stacs}.
This research was supported in part by the Swedish Research Council, Grant VR 2012-4730 (A.B.),
the Academy of Finland,
Grants 252083 and 256287 (P.K.), and by the National Science Centre
of Poland, Grant N206 567140 (\L.K.).
The third author thanks Sylwia Antoniuk, Marek Cygan, Michal Debski, and Matthias Mnich for helpful discussions on related topics.

\bibliographystyle{abbrv}

\end{document}